\documentclass[letterpaper, 12pt]{article}

\usepackage{graphicx}

\usepackage{amsmath}

\newcommand{\comment}[1]{}

\newcommand{\shortdividerline}{
\begin{center} \line(1,0){150} \end{center}
}
\newcommand{\dividerline}{\begin{center}\hrule\end{center}}

\usepackage{graphicx}
\usepackage{latexsym}
\newcommand{\scriptf}{\mathcal{F}}
\newcommand{\scripte}{\mathcal{E}}
\newcommand{\scriptv}{\mathcal{V}}

\newenvironment{proof}{\paragraph{\bf Proof:}}{\hspace*{\fill}\(\Box\)}
\newtheorem{theorem}{Theorem}

\newtheorem{claim}{Claim}
\newtheorem{corollary}{Corollary}
\newtheorem{definition}{Definition}

\newtheorem{lemma}{Lemma}

\def\noflash#1{\setbox0=\hbox{#1}\hbox to 1\wd0{\hfill}}

\begin{document}

\title{Iterative Approximate Byzantine Consensus\\ in Arbitrary Directed Graphs
 \footnote{\normalsize This research is supported
in part by
 National
Science Foundation award CNS 1059540. Any opinions, findings, and conclusions or recommendations expressed here are those of the authors and do not
necessarily reflect the views of the funding agencies or the U.S. government.}}

\author{Nitin Vaidya$^{1,3}$, Lewis Tseng$^{2,3}$, and Guanfeng Liang$^{1,3}$\\~\\
 \normalsize $^1$ Department of Electrical and Computer Engineering, \\
 \normalsize $^2$ Department of Computer Science,
 and\\ \normalsize $^3$ Coordinated Science Laboratory\\ \normalsize University of Illinois at Urbana-Champaign\\ \normalsize Email: \{nhv, ltseng3, gliang2\}@illinois.edu~\\~\\Technical Report\thanks{ \normalsize This report is a modified version of a previous technical report (Nitin Vaidya, Lewis Tseng, and Guanfeng Liang. Iterative Approximate Byzantine Consensus 
in Arbitrary Directed Graphs. CoRR, abs/1201.4183, January 19 2012. http://arxiv.org/abs/1201.4183)}}

\date{February 2, 2012}
\maketitle

%\noindent
%Guanfeng Liang is a full time Ph.D student in University of Illinois at Urbana-Champaign. This paper should be considered for the best student paper award.

\thispagestyle{empty}

\newpage

\section{Introduction}
\label{s_intro}

In this paper, we explore the problem of iterative approximate Byzantine consensus in arbitrary directed graphs. In particular, we prove a necessary and sufficient condition for the existence of iterative Byzantine consensus algorithms. Additionally, we use our sufficient condition to examine whether such algorithms exist for some specific graphs.

Approximate Byzantine consensus \cite{AA_Dolev_1986} is a natural extension of original Byzantine Generals (or Byzantine consensus) problem \cite{psl_BG_1982}. 
The goal in approximate consensus is to allow the fault-free nodes
to agree on values that are approximately equal to each other.
There exist iterative algorithms for the approximate consensus problem that work correctly in fully connected graphs \cite{AA_Dolev_1986,AA_nancy} when the number of
nodes $n$ exceeds $3f$, where $f$ is the upper bound on the number of failures.
 In \cite{AA_Fekete_aoptimal}, Fekete studies the convergence rate of approximate consensus
algorithms. 
% proved a lower bound on the convergence rate of iterative approximate consensus, and proposed an asymptotically optimal approximate agreement algorithm, when $n > 4f$.
 % Recently, instead of using an iterative strategy, Ben-Or et al. adopted a simple algorithm based on Gradecast to solve approximate consensus efficiently in a fully connected network \cite{gradecast_benor_2010}. Their solution has early-stopping property and is optimal in convergence rate, number of nodes, and message complexity. 
 Ben-Or et al. develop an algorithm based on Gradcast to solve approximate consensus efficiently in a fully connected network \cite{gradecast_benor_2010}.
% Their solution has early-stopping property and is optimal in convergence rate, number of nodes, and message complexity. 

There have been attempts at achieving approximate consensus iteratively
in partially connected graphs. In \cite{AA_PCN_Local}, Kieckhafer and Azadmanesh examined the necessary conditions in order to achieve ``local'' convergence and performed a case study on global convergence in some special graphs. Later, they extended their work to asynchronous systems \cite{AA_async_PCN}. In \cite{AA_PFCN}, Azadmanesh et al. showed how to build a special network, called Partially Fully Connected Network, in which global convergence is achieved. Srinivasan and Azadmanesh studied the application of iterative approximate consensus in data aggregation, and developed an analytical approach using Markov chains \cite{DA_PCN_Markov,DA_PCN_Markov_conf}.
% They suggest that if the network has only bi-directional links and has at least one self-loop link, then global convergence is ensured based on the properties of Discrete Markov Chains (\cite{DA_PCN_Markov_conf} and \cite{DA_PCN_Markov}). 
% However, we demonstrate that these conditions alone are insufficient.  

In \cite{DFC_Byzantine}, Sundaram and Hadjicostis explored Byzantine-fault tolerant distributed function calculation in an arbitrary
network assuming a {\em broadcast model}. Under the {\em broadcast}
model, every transmission of a node is received by all its neighbors.
Hence, faulty nodes can send false data, but they have to send exactly the same piece of data to all their neighbors. They proved that distributed function calculation is possible if network connectivity is at least $2f+1$. Their algorithm maintains more ``history'' (a sequence of previous states) than the iterative algorithms considered in this paper. 

In \cite{IBA_broadcast_Sundaram}, Zhang and Sundaram studied the sufficient conditions for iterative consensus algorithm under ``f-local" fault model. They also provided a construction of graphs satisfying the sufficient conditions.

LeBlanc and Koutsoukos \cite{leblanc1} address a continuous time version of the Byzantine consensus problem in {\em complete} graphs.
Recently, for the {\em broadcast} model, 
LeBlanc et al. have independently developed necessary and sufficient conditions for $f$-fault tolerant approximate consensus in arbitrary graphs \cite{diss_Sundaram}; in \cite{leblanc2} they have developed some {\em sufficient}
conditions for correctness of a
class of iterative consensus algorithms.

 To the best of our knowledge, characterization of {\em tight necessary and sufficient} conditions for iterative approximate consensus in arbitrary directed graphs in the presence of Byzantine faults under point-to-point model is still an open problem. Iterative approximate consensus algorithms without any fault tolerance capability (i.e., $f=0$) in arbitrary graphs have been explored extensively. The proof of convergence presented in this paper is inspired by the prior work on non-fault-tolerant algorithms \cite{AA_convergence_markov}.

\section{Preliminaries}
\label{sec:models}

\subsection{Network Model}
The network is modeled as a simple {\em directed} graph $G(\scriptv,\scripte)$, where $\scriptv=\{1,\dots,n\}$ is the set of $n$ nodes, and $\scripte$ is the set of directed edges between nodes in $\scriptv$. We use the terms ``edge'' and ``link'' interchangeably.
We assume that $n\geq 2$, since the consensus problem for $n=1$ is trivial.
If a directed edge $(i,j)\in \scripte$, then node $i$ can reliably transmit to node $j$.
For convenience,
we exclude self-loops from $\scripte$, although every node is allowed to send messages
to itself. We also assume that all edges are authenticated, such that when a node $j$ receives a message from node $i$ (on edge $(i,j)$), it can correctly determine that the message was sent by node $i$.
For each node $i$, let $N_i^-$ be the set of nodes from which $i$ has incoming
edges.
That is, $N_i^- = \{\, j ~|~ (j,i)\in \scripte\, \}$.
Similarly, define $N_i^+$ as the set of nodes to which node $i$
has outgoing edges. That is, $N_i^+ = \{\, j ~|~ (i,j)\in \scripte\, \}$.
By definition, $i\not\in N_i^-$ and $i\not\in N_i^+$.
However, we emphasize that each node can indeed send messages to itself.
The network is assumed to be synchronous.

\subsection{Failure Model}

We consider the Byzantine failure model, with up to $f$ nodes becoming faulty. A faulty node may misbehave arbitrarily. Possible misbehavior includes sending incorrect and mismatching messages to different neighbors. The faulty nodes may potentially collaborate with each other. Moreover, the faulty nodes are assumed to have a complete knowledge of the state of the other nodes in the system and a complete knowledge of specification of the algorithm.

\subsection{Iterative Approximate Byzantine Consensus}

We consider iterative Byzantine consensus as follows:
\begin{itemize}
\item Up to $f$ nodes in the network may be Byzantine faulty.
\item Each node starts with an {\em input}, which is assumed to be a single
real number.
\item Each node $i$ maintains state $v_i$, with $v_i[t]$ denoting the state
of node $i$ at the end of the $t$-th iteration of the algorithm.
$v_i[0]$ denotes the initial state of node $i$, which is set equal to its
{\em input}.
Note that, at the start of the $t$-th iteration ($t>0$), the state of
node $i$ is $v_i[t-1]$.
\item The goal of an approximate consensus algorithm is
to allow each node to compute an {\em output} in {\em each iteration}
with the following two properties:
\begin{itemize}
\item {\bf Validity:} The output of each node is within the convex hull of the
inputs at the {\em fault-free}\, nodes.
\item {\bf Convergence:}
The outputs of the different fault-free nodes converge to an identical value
as $t\rightarrow \infty$.
\end{itemize}
\item {\bf Output constraint:} For the family of iterative algorithms considered in this paper,
output of node $i$
at time $t$ is {\bf equal to} its state $v_i[t]$.
\end{itemize}
The iterative algorithms will be implemented as follows:
\begin{itemize}
\item At the start of $t$-th iteration, $t\ge 1$, each node $i$ sends $v_i[t-1]$ on all its
outgoing links (to nodes in $N_i^+$).

\item Denote by $r_i[t]$ the vector of values received by node $i$ from
nodes in $N_i^-$ at time $t$. The size of vector $r_i[t]$ is $|N_i^-|$.
\item Node $i$ updates its state using some transition function $Z_i$ as
follows,
where $Z_i$ is part of the specification of the algorithm:
\[ v_i[t] = Z_i(r_i[t],v_i[t-1]) \]

Since the inputs are real numbers, and because we impose the above {\em output
constraint}, the state of each node in each iteration is also viewed as
a real number.
\end{itemize}
The function $Z_i$ may be dependent on the network topology. However,
as seen later, for convergence, it suffices for each node $i$ to know 
$N_i^-$.

Observe that, given the state of the nodes at time $t-1$, their state at time
$t$ is independent of the prior history. The evolution of the state of the nodes
may, therefore, be modeled by a Markov chain (although we will not use that
approach in this paper).

We now introduce some notations.
\begin{itemize}
\item Let $\scriptf$ denote the set of Byzantine faulty nodes, where $|\scriptf|\leq f$.
Thus, the set of fault-free nodes is $\scriptv-\scriptf$.
\footnote{For sets $X$ and $Y$, $X-Y$ contains elements
that are in $X$ but not in $Y$. That is,
$X-Y=\{i~|~ i\in X,~i\not\in Y\}$.
}
\item $U[t] = \max_{i\in\scriptv-\scriptf}\,v_i[t]$. $U[t]$
is the largest state among the fault-free nodes (at time $t$).
Recall that, due
to the {\em output constraint}, the state
of node $i$ at the end of iteration $t$ (i.e., $v_i[t]$) is also its output in
iteration $t$.
\item $\mu[t] = \min_{i\in\scriptv-\scriptf}\,v_i[t]$. $\mu[t]$ is the smallest state
among the fault-free nodes at time $t$ (we will use the phrase ``at time $t$''
interchangeably with ``at the end of $t$-th iteration'').
\end{itemize}

With the above notation, we can restate the {\em validity} and
{\em convergence} conditions as follows:
\begin{itemize}
\item {\bf Validity:} $\forall t>0$, $U[t]\le U[0]$ and $\mu[t]\ge \mu[0]$
\item {\bf Convergence:} $\lim_{\,t\rightarrow\infty} ~ U[t]-\mu[t] = 0$
\end{itemize}

% ++++++++++++ is paragraph below convincing enough? ++++++++++++++++

The {\em output constraint} and the {\em validity}
condition together imply that the iterative algorithms of interest
do not maintain a ``sense of time''. In particular, the iterative
computation by the algorithm, as captured in functions $Z_i$, cannot
explicitly take the elapsed time (or $t$) into account.\footnote{In a practical 
implementation, the algorithm may keep track of time, for instance, to
decide to terminate after a certain number of iterations.}
Due to this, the validity condition for algorithms of interest here
becomes:
\begin{eqnarray}
\noindent
\mbox{\bf Validity:}~~~~~\forall t>0,~U[t]\le U[t-1]~\mbox{and}~\mu[t]\ge \mu[t-1]
\hspace*{2.5in}~
\label{e_validity}
\end{eqnarray}
In the discussion below, when we refer to the validity condition,
we mean (\ref{e_validity}).

For illustration, below we present Algorithm 1 that satisfies
the {\em output constraint}.
The algorithm has been proved to achieve {\em validity}
and {\em convergence} in fully connected
graphs with $n>3f$ \cite{AA_Dolev_1986, AA_nancy}. We will later address correctness of this
algorithm in arbitrary graphs.\\ Here, we assume that each node $v \in \scriptv$ has at least $2f$ incoming links. That is $|N_i^-| \geq 2f$. Later, we will show that there is no iterative Byzantine consensus if this condition does not hold.

\dividerline

\noindent
{\bf Algorithm 1}\\
Steps that should be performed by each node $i\in \scriptv$ in the
$t$-th iteration are as follows. Note that the faulty nodes may
deviate from this specification.
Output of node $i$ at time $t$ is $v_i[t]$.
\begin{enumerate}
\item Transmit current state $v_i[t-1]$ on all outgoing edges.
\item Receive values on all incoming edges (these values form
vector $r_i[t]$ of size $|N_i^-|$).
\item Sort the values in $r_i[t]$ in an increasing order, and eliminate
the smallest $f$ values, and the largest $f$ values (breaking ties
arbitrarily). Let $N_i^*[t]$ denote the identifiers of nodes from
whom the remaining $N_i^- - 2f$ values were received, and let
$w_j$ denote the value received from node $j\in N_i^*$.
Then, $|N_i^*[t]| = |N_i^-| - 2f$.
By definition, $i\not\in N_i^*[t]$.
Note that if $j\in \{i\}\cup N_i^*[t]$ is fault-free, then $w_j=v_j[t-1]$.
Define
\begin{eqnarray}
v_i[t] = Z_i(r_i[t],v_i[t-1]) =\sum_{j\in \{i\}\cup N_i^*[t]} a_i \, w_j
\label{e_Z}
\end{eqnarray}
where
\[ a_i = \frac{1}{|N_i^-|+1-2f}
\] 
% \shortdividerline

The ``weight'' of each term on the right side of
(\ref{e_Z}) is $a_i$, and these weights add to 1.
Also, $0<a_i\leq 1$.
For future reference, let us define $\alpha$ as:
\begin{eqnarray}
\alpha = \min_{i\in \scriptv}~a_i
\label{e_alpha}
\end{eqnarray}

\end{enumerate}

\dividerline

\section{Necessary Condition}
\label{s_necessary}

% $n\geq 3f+1$ is a necessary condition for achieving approximate agreement
% in the presence of up to $f$ Byzantine faults \cite{who}.
% We will refer to this necessary condition as NC1.

For an iterative Byzantine approximate consensus algorithm satisfying the
{\em output constraint}, the {\em validity} condition,
 and the {\em convergence} condition
 to exist, the underlying graph $G(\scriptv,\scripte)$
must satisfy a necessary condition proved in this section.
We now define relations $\Rightarrow$
and $\not\Rightarrow$ that are used frequently in our proofs.

\begin{definition}
\label{def:absorb}
For non-empty disjoint sets of nodes $A$ and $B$,
$A \Rightarrow B$ iff there exists a node $v\in B$ that has at least $f+1$ incoming
links from nodes in $A$, i.e., $|N_v^-\cap A|>f$.\\
$A\not\Rightarrow B$ iff $A\Rightarrow B$ is {\em not} true.
\end{definition}

\dividerline

\begin{theorem}
\label{thm:nc}
Let sets $F,L,C,R$ form a partition\footnote{Sets $X_1,X_2,X_3,...,X_p$
are said to form a partition of set $X$ provided that (i) $\cup_{1\leq i\leq p} X_i = X$
and $X_i\cap X_j=\Phi$ when $i\neq j$.} of $\scriptv$, such that
\begin{itemize}
\item $0 \leq |F|\le f$,
% \item $0 \leq |C|\le f$,
\item $0<|L|$, and
\item $0<|R|$
\end{itemize}
Then, at least one of the two conditions below must be true.
\begin{itemize}
\item $C\cup R\Rightarrow L$
\item $L\cup C\Rightarrow R$
\end{itemize}
\end{theorem}
\begin{proof}
The proof is by contradiction.
Let us assume that a correct iterative consensus algorithm exists,
and $C\cup R\not\Rightarrow L$ and $L\cup C\not\Rightarrow R$.
Thus, for any $i\in L$, $|N_i^-\cap (C\cup R)|<f+1$,
and 
$j\in R$, $|N_j^-\cap (L\cup C)|<f+1$,
Figure~\ref{f_thm1} illustrates the sets used in this proof.

\begin{figure}
\centering
\includegraphics[width=0.7\textwidth]{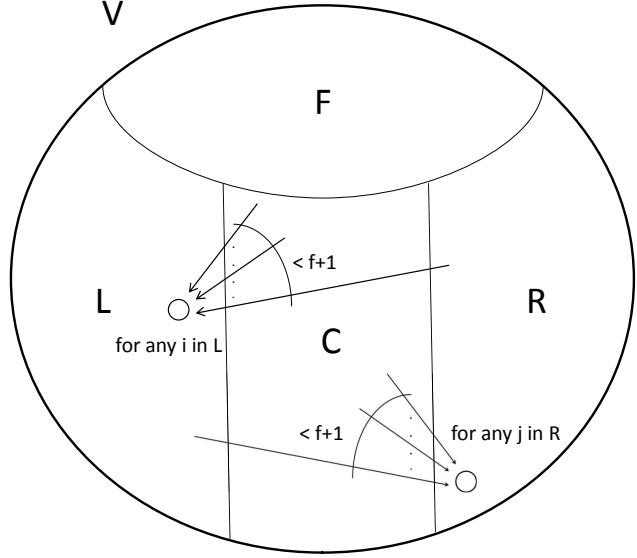}
\caption{Illustration for the proof of Theorem \ref{thm:nc}. In this figure, $C\cup R \not\Rightarrow L$ and $L\cup C \not\Rightarrow R$.}
\label{f_thm1}
\end{figure}

Also assume that the nodes in $F$ (if $F$ is non-empty) are all faulty,
and the remaining nodes, in sets $L,R,C$, are fault-free. Note that the fault-free nodes
are not necessarily aware of the identity of the faulty nodes.

Consider the case when (i) each node in $L$ has input $m$, (ii) each
node in $R$ has input $M$, such that $M>m$,
and (iii) each node in $C$, if $C$ is non-empty,
has an input in the range
$[m,M]$.

At the start of iteration 1, suppose that the faulty nodes in $F$ (if non-empty)
send $m^- < m$ to nodes in $L$, send $M^+ > M$ to nodes in $R$, and
send some arbitrary value in $[m,M]$ to the nodes in $C$ (if $C$ is
non-empty).
This behavior is possible since nodes in $F$ are faulty.
Note that $m^-<m<M<M^+$.
Each fault-free node $k\in\scriptv-\scriptf$, sends to nodes
in $N_k^+$ value $v_k[0]$ in iteration 1.

Consider any node $i \in L$. Denote $N'(i) = N_i^- \cap (C\cup R)$.
Since $C\cup R\not\Rightarrow L$, $|N'(i)|\le f$.
Node $i$ will then receive $m^-$ from the nodes in $F \cap N_i^-$,
and values in $[m,M]$ from the nodes in $N'(i)$, and
$m$ from the nodes in $\{i\}\cup (L\cap N_i^-)$.

Consider four cases:
\begin{itemize}
\item $F$ and $N'(i)$ are both empty:
In this case, all the values that $i$ receives are from nodes in $\{i\}\cup( L\cap N_i^-)$,
and are identical to $m$. By validity condition (\ref{e_validity}), node $i$ must set its
new state, $v_i[1]$, to be $m$ as well.

\item $F$ is empty and $N'(i)$ is non-empty:
In this case, since $|N'(i)|\leq f$, from $i$'s perspective,
it is possible that all the nodes in
$N'(i)$ are faulty, and the rest of the nodes are fault-free. In this
situation, the values sent to node $i$ by the fault-free nodes (which are
all in $\{i\}\cup (L\cap N_i^-))$ are all $m$, and therefore, $v_i[1]$ must be set to $m$
as per the validity condition (\ref{e_validity}).

\item $F$ is non-empty and $N'(i)$ is empty:
In this case, since $|F|\leq f$, it is possible that all the nodes in $F$ are faulty,
and the rest of the nodes are fault-free. In this
situation, the values sent to node $i$ by the fault-free nodes (which are
all in $\{i\}\cup (L\cap N_i^-))$ are all $m$, and therefore, $v_i[1]$ must be set to $m$
as per the validity condition (\ref{e_validity}).

\item Both $F$ and $N'(i)$ are non-empty:
From node $i$'s perspective, consider two possible scenarios:
(a) nodes in $F$ are faulty, and the other
nodes are fault-free, and (b) nodes in $N'(i)$ are faulty, and the
other nodes are fault-free.

In scenario (a), from node $i$'s perspective, the non-faulty nodes have values
in $[m,M]$ whereas the faulty nodes have value $m^-$. According to the validity
condition (\ref{e_validity}), $v_i[1] \geq m$. On the other hand, in scenario (b), the
non-faulty nodes have values $m^-$ and $m$, where $m^-<m$; so $v_i[1] \leq m$, according to
the validity condition (\ref{e_validity}). Since node $i$ does not know whether the
correct scenario is (a) or (b), it must update its state to satisfy the
validity condition in both cases. Thus, it follows that $v_i[1] = m$.
\end{itemize}
Observe that in each case above $v_i[1]=m$ for each node $i\in L$.
Similarly, we can show that $v_j[1] = M$ for each node $j \in R$.

Now consider the nodes in set $C$, if $C$ is non-empty.
All the values received by the nodes in $C$ are in $[m,M]$, therefore,
their new state must also remain in $[m,M]$, as per the validity condition.

The above discussion implies that, at the end of the first iteration,
the following conditions hold true: (i) state of each node in $L$ is
$m$, (ii) state of each node in $R$ is $M$, and (iii) state of each node
in $C$ is in $[m,M]$. These conditions are identical to the initial conditions
listed previously. Then, by induction, it follows that for
any $t \geq 0$, $v_i[t] = m, \forall i \in L$, and $v_j[t] = M, \forall j \in R$.
Since $L$ and $R$ contain fault-free nodes, the convergence requirement
is not satisfied. This is a contradiction to the assumption that a correct
iterative algorithm exists.
\end{proof}

\dividerline

\begin{corollary}
\label{cor:nc2}
Let $\{F,L,R\}$ be a partition of $\scriptv$, such that $0\leq |F|\le f$, and
$L$ and $R$ are non-empty. Then, either $L\Rightarrow R$ or $R\Rightarrow L$.
\end{corollary}
\begin{proof}
The proof follows by setting $C=\Phi$ in Theorem~\ref{thm:nc}.
\end{proof}

\dividerline

While the two corollaries below are also proved in prior literature \cite{impossible_proof_lynch},
we derive them again using the necessary condition above.
\begin{corollary}
\label{cor:3f}
The number of nodes $n$ must exceed $3f$ for
the existence of a correct iterative consensus algorithm tolerating $f$ failures.
\end{corollary}
\begin{proof}
The proof is by contradiction.
Suppose that $2\leq n\leq 3f$, and consider the following two cases:
\begin{itemize}
\item $2\leq n\leq 2f$: Suppose that $L,R,F$ is a partition of $\scriptv$
such that $|L|=\lceil n/2 \rceil\leq f$,
$|R|=\lfloor n/2 \rfloor\leq f$ and $F=\Phi$. Note that $L$ and $R$
are non-empty, and $|L|+|R|=n$.
\item $2f<n\leq 3f$:
% Since $f$ is an integer and $f\geq n/3$, it follows that $f\geq \lceil n/3 \rceil$. 
Suppose that $L,R,F$ is a partition of $\scriptv$,
such that $|L|=|R|=f$ and $|F|=n-2f$. Note that
$0<|F|\leq f$.
\end{itemize}
In both cases above, Corollary~\ref{cor:nc2} is applicable. Thus,
either $L\Rightarrow R$ or $R\Rightarrow L$.
For $L\Rightarrow R$ to be true, $L$ must contain at least $f+1$ nodes.
Similarly, for $R\Rightarrow L$ to be true, $R$ must contain at least
$f+1$ nodes. Therefore, at least one of the sets $L$ and $R$ must contain more than
$f$ nodes. This contradicts our choice of $L$ and $R$ above
(in both cases, size of $L$ and $R$ is $\leq f$).
Therefore, $n$ must be larger than $3f$. 
\end{proof}

\dividerline

% +++++++++++++++ it is a bit strange that the corollary below doesn't hold
% for $f=0$ ---- need to see if there is any inconsistency someplace ++++++++++++++++

\begin{corollary}
\label{cor:2f+1}
When $f>0$, for each node $i\in\scriptv$,
$|N_i^-|\geq 2f+1$, i.e., each node $i$
has at least $2f+1$ incoming links.
\end{corollary}
\begin{proof}
The proof is by contradiction. Suppose that for some node $i$,
$|N_i^-|\leq 2f$.
Define set $L=\{i\}$.
Partition $N_i^-$ into two sets $F$ and $H$
such that $|H|=\lfloor |N_i^-|/2\rfloor\leq f$
and $|F|=\lceil |N_i^-|/2\rceil\leq f$.
Define $R=V-F-L=V-F-\{i\}$.
Thus, $N_i^-\cap R=H$, and $|N_i^-\cap R|\leq f$.
Therefore, since $L=\{i\}$ and  $|N_i^-\cap R|\leq f$, $R\not\Rightarrow L$.
Also, since $|L|=1<f+1$, $L\not\Rightarrow R$.

This violates Corollary~\ref{cor:nc2}.
\end{proof}

\section{Useful Lemmas}
\label{s_useful}

\begin{definition}
For disjoint sets $A,B$, 
$in(A \Rightarrow B)$ denotes the set of
all the nodes in $B$ that each have at least $f+1$ incoming links from
nodes in $A$. More formally,
\[
in(A\Rightarrow B) = \{~v~|v\in B \mbox{\normalfont~and~}~f+1\leq |N_v^-\cap A|~\}
\]

With a slight abuse of notation, when $A\not\Rightarrow B$, define $in(A\Rightarrow B)=\Phi$.
\end{definition}

\comment{================ do we use the definition of $in$ above when $A\not\Rightarrow B$ ?
This is OK, but need to clarify that we can apply the above definition even in those
cases. For instance, ``when $A\not\Rightarrow B$, we have
$in(A\Rightarrow B)=\Phi$''.  A slight abuse of notation here.  =======================}

\shortdividerline

\begin{definition}
\label{def:absorb_sequence}
For {\em non-empty disjoint} sets $A$ and $B$, set $A$ is said to {\bf propagate to} set $B$ in $l$ steps, where $l>0$,
if there exist sequences of sets $A_0,A_1,A_2,\cdots,A_l$ and $B_0,B_1,B_2,\cdots,B_l$ (propagating sequences) such that

\begin{itemize}
\item $A_0=A$, $B_0=B$, $B_l=\Phi$, and, for $\tau<l$, $B_\tau \neq \Phi$.
\item for $0\leq \tau\leq l-1$,
\begin{list}{}{}
\item[*] $A_\tau\Rightarrow B_\tau$,
\item[*] $A_{\tau+1} = A_\tau\cup in(A_\tau\Rightarrow B_\tau)$, and
\item[*] $B_{\tau+1} = B_\tau - in(A_\tau\Rightarrow B_\tau)$
\end{list}
\end{itemize}
\end{definition}
Observe that $A_\tau$ and $B_\tau$ form a partition of $A\cup B$,
and for $\tau<l$, $in(A_\tau\Rightarrow B_\tau)\neq \Phi$.
Also, when set $A$ propagates to set $B$, length $l$ above is
necessarily finite. In particular, $l$ is upper bounded by $n-f-1$, since set $A$ must
be of size at least $f+1$ for it to propagate to $B$.
\dividerline

\begin{lemma}
\label{lemma:absorb_condition}
Assume that $G(\scriptv,\scripte)$ satisfies Theorem~\ref{thm:nc}.
Consider a partition $A,B,F$ of $\scriptv$ such that
$A$ and $B$ are non-empty, and $|F|\leq f$.
If $B \not\Rightarrow A$, then set $A$ propagates to set $B$.
\end{lemma}

\begin{proof}
Since $A,B$ are non-empty, and $B\not\Rightarrow A$, by
Corollary~\ref{cor:nc2}, we have $A\Rightarrow B$.

The proof is by induction.
Define $A_0=A$ and $B_0=B$.
Thus $A_0\Rightarrow B_0$ and $B_0\not\Rightarrow A_0$.
Note that $A_0$ and $B_0$ are non-empty.
% Define $A_1=A_0\cup in(A_0\Rightarrow B_0)$
% and $B_1 = B_0 - in(A_0\Rightarrow B_0)$.
% If $B_1=\Phi$, then by Definition~\ref{def:absorb_sequence},
% $A$ propagates to $B$, and the proof is complete.
% Hereafter, consider the case when $B_1\neq\Phi$.
% In this case, we will prove the lemma by induction. 

\noindent
{\em Induction basis}: For some $\tau\geq 0$,
\begin{itemize}
\item for $0\leq k < \tau$, $A_k\Rightarrow B_k$, and $B_k\neq \Phi$,
\item either $B_\tau=\Phi$ or $A_\tau\Rightarrow B_\tau$,
\item for $0\leq k< \tau$,
$A_{k+1} = A_k \cup in(A_k\Rightarrow B_k)$, and
$B_{k+1}=B_k - in(A_k\Rightarrow B_k)$
\end{itemize}
Since $A_0\Rightarrow B_0$, 
the induction basis holds true for $\tau=0$.
 
\noindent
{\em Induction:}
If $B_\tau=\Phi$, then the proof is complete, since all
the conditions specified in Definition~\ref{def:absorb_sequence} are satisfied
by the sequences of sets $A_0,A_1,\cdots,A_\tau$ and $B_0,B_1,\cdots,B_\tau$.

\begin{figure}[h]
\centering
\includegraphics[width=0.7\textwidth]{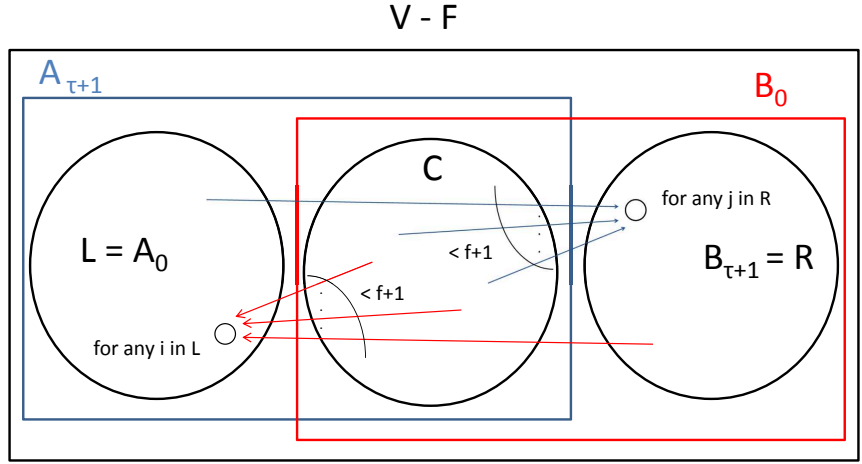}
\caption{Illustration for the proof Lemma \ref{lemma:absorb_condition}. In this figure, $B_0\not\Rightarrow A_0$ and $A_{\tau+1}\not\Rightarrow B_{\tau+1}$}.
\label{f_lemma1}
\end{figure}

Now consider the case when $B_\tau\neq \Phi$. By assumption,
$A_k\Rightarrow B_k$, for $0\leq k\leq \tau$.
Define $A_{\tau+1} = A_\tau \cup in(A_\tau\Rightarrow B_\tau)$ and $B_{\tau+1}=B_\tau - in(A_\tau\Rightarrow B_\tau)$.
Our goal is to prove that either $B_{\tau+1}=\Phi$ or $A_{\tau+1}\Rightarrow B_{\tau+1}$.
If $B_{\tau+1}=\Phi$, then the induction is complete. Therefore, now let us assume
that $B_{\tau+1}\neq \Phi$ and prove that $A_{\tau+1}\Rightarrow B_{\tau+1}$.
 We will prove this by contradiction.

Suppose that $A_{\tau+1}\not\Rightarrow B_{\tau+1}$.
Define subsets $L,C,R$ as follows: $L=A_0$, $C=A_{\tau+1}-A_0$ and $R=B_{\tau+1}$. Figure~\ref{f_lemma1} illustrates the sets used in this proof. 
Due to the manner in which $A_k$'s and $B_k$'s
are defined, we also have $C=B_0-B_{\tau+1}$.
Observe that $L,C,R,F$ form a partition of $\scriptv$, where $L,R$ are
non-empty, and the following relationships hold:
\begin{itemize}
\item $C\cup R = B_0$, and
\item $L\cup C=A_{\tau+1}$
\end{itemize}
Rewriting $B_0\not\Rightarrow A_0$ and $A_{\tau+1}\not\Rightarrow B_{\tau+1}$,
using the above relationships, we have, respectively,
\[
C\cup R\not\Rightarrow L,
\]
and
\[
L\cup C\not\Rightarrow R
\]
This violates the necessary condition in Theorem~\ref{thm:nc}.
This is a contradiction, completing the induction.

Thus, we have proved that, either (i) $B_{\tau+1}=\Phi$,
or (ii) $A_{\tau+1}\Rightarrow B_{\tau+1}$.
Eventually, for large enough $t$, $B_t$ will become $\Phi$, resulting
in the propagating sequences $A_0,A_1,\cdots, A_t$ and
$B_0,B_1,\cdots,B_t$, satisfying the conditions in Definition~\ref{def:absorb_sequence}.
Therefore, $A$ propagates to $B$.

\end{proof}

\dividerline

\begin{lemma}
\label{lemma:must_absorb}
Assume that $G(\scriptv,\scripte)$ satisfies Theorem~\ref{thm:nc}.
For any partition $A,B,F$ of $\scriptv$, where $A,B$ are both non-empty,
and $|F|\leq f$,
at least one of the following conditions must be true:
\begin{itemize}
\item $A$ propagates to $B$, or
\item $B$ propagates to $A$
% \item there exist sets $L$ and $R$ such that (i) $A\subseteq L$,
% (ii) $R\subseteq B$, (iii) $L,R$ is a partition of $\scriptv-F$,
% and (iv) $R$ propagates to $L$.
\end{itemize}
\end{lemma}

\begin{proof}
Consider two cases:
\begin{itemize}
\item $A\not\Rightarrow B$: Then by Lemma \ref{lemma:absorb_condition},
$B$ propagates to $A$, completing the proof.

\item $A\Rightarrow B$: In this case, consider two sub-cases:
\begin{itemize}
\item $A$ propagates to $B$: The proof in this case is complete.

\item $A$ does not propagate to $B$:
Thus, propagating sequences defined in Definition~\ref{def:absorb_sequence}
do not exist in this case. More precisely, there must exist $k>0$,
and sets $A_0,A_1,\cdots,A_k$ and $B_0,B_1,\cdots,B_k$,
such that:
\begin{itemize}
\item $A_0=A$ and $B_0=B$, and
\item for $0\leq i\leq k-1$,
\begin{list}{}{}
\item[o] $A_i\Rightarrow B_i$,
\item[o] $A_{i+1} = A_i\cup in(A_i\Rightarrow B_i)$, and
\item[o] $B_{i+1} = B_i - in(A_i\Rightarrow B_i)$.
\end{list}
\item $B_{k}\neq \Phi$ and $A_{k}\not\Rightarrow B_{k}$.
\end{itemize}
The last condition above violates the requirements for $A$ to propagate
to $B$.
% ++++++++++++ need to check above claims carefully +++++++++++++++++++++

Now $A_{k}\neq \Phi$, $B_k\neq \Phi$, and $A_k,B_k,F$ form
a partition of $\scriptv$. Since $A_{k}\not\Rightarrow B_{k}$,
by Lemma \ref{lemma:absorb_condition},
$B_k$ propagates to $A_k$.
% Define $L=A_k$ and $R=B_k$. Thus, $R$ propagates to $L$.
% Note that $A_0\subseteq A_k$,
% and $B_k\subseteq B_0$, therefore, $A = A_0\subseteq L$
% and $R\subseteq B_0=B$.
% By construction of $A_k=L$ and $B_k=R$, we have that $L,R$ is a partition
% of $\scriptv-F$.
% Thus, the lemma is proved.

Since $B_k\subseteq B_0 = B$, $A\subseteq A_k$, and $B_k$ propagates
to $A_k$, it should be easy to see that $B$ propagates to $A$. The proof is presented in the Appendix \ref{app:lemma:must_absorb} for completeness.

% It should be easy to see that that $B$ propagates to $A$.
% Let us define $P=P_0=B_k$ and $Q=Q_0=A_k$. Thus, $P$ propagates to $Q$.
% Note that $P_0=B_k\subseteq B_0$ and $A_0\subseteq A_k =Q_0$.
% We will define sets additional sets $P_i,Q_i$ below.
% \begin{itemize}
% \item
% Define $P_1 = P_0\cup (in(P_0\Rightarrow Q_0))$. Also,
% $B_1 = B_0\cup (in(B_0\Rightarrow A_0))$.
% Observe that $in(P_0\Rightarrow Q_0) \subseteq in(B_0\Rightarrow A_0)$.
% Also, $P_0\subseteq B_0$ and $Q_0\subseteq A_0$.
% Therefore, $P_1\subseteq B_1$.
% \item If $P_1$ is
% \end{itemize}
% 
\end{itemize}
\end{itemize}
\end{proof}

\dividerline

\section{Sufficiency}
\label{s_sufficiency}

We prove that the necessary condition in Theorem~\ref{thm:nc} is sufficient.
In particular, we will prove that Algorithm 1 satisfies
{\em validity} and {\em convergence} conditions when the necessary condition
is satisfied.

In the discussion below, assume that graph $G(\scriptv,\scripte)$ satisfies
Theorem~\ref{thm:nc}, and that $\scriptf$ is the set of faulty nodes in the network.
Thus, the nodes in $\scriptv-\scriptf$ are fault-free.  
Since Theorem~\ref{thm:nc} holds for $G(\scriptv,\scripte)$,
all the subsequently developed
corollaries and lemmas in Sections~\ref{s_necessary} and \ref{s_useful}
also hold for $G(\scriptv,\scripte)$.

\begin{theorem}
\label{thm:validity}
Suppose that $G(\scriptv, \scripte)$ satisfies Theorem~\ref{thm:nc}.
Then Algorithm 1 satisfies the {\em validity} condition (\ref{e_validity}).
\end{theorem}
\begin{proof}

Consider the $t$-th iteration, and any fault-free node $i\in\scriptv-\scriptf$.
Consider two cases:
\begin{itemize}
\item
$f=0$: In (\ref{e_Z}), note that $v_i[t]$ is computed using
states from the previous iteration at node $i$ and other nodes.
By definition of $\mu[t-1]$ and $U[t-1]$, $v_j[t-1]\in [\mu[t-1],U[t-1]]$
for all fault-free nodes $j\in\scriptv-\scriptf$.
Thus, in this case, all the values used in computing $v_i[t]$ are in the
range $[\mu[t-1],U[t-1]]$.
Since $v_i[t]$ is computed as a 
weighted average of these values, $v_i[t]$ is also within
$[\mu[t-1],U[t-1]]$.
\item $f>0$: By Corollary~\ref{cor:2f+1}, $|N_i^-|\geq 2f+1$, and therefore,
$|r_i[t]| \geq 2f+1$.
When computing set $N_i^*[t]$, the largest $f$ and smallest $f$ values
from $r_i[t]$ are eliminated. Since at most $f$ nodes are faulty,
it follows that, either (i) the values received from the faulty
nodes are all eliminated, or (ii) the values from the
faulty nodes that still remain are between values
received from two fault-free nodes. Thus, the remaining values in $r_i[t]$ ($v_j[t-1],~\forall j \in N_i^*[t]$) are
all in the range $[\mu[t-1],U[t-1]]$. Also, $v_i[t-1]$ is 
in $[\mu[t-1],U[t-1]]$, as per the definition of $\mu[t-1]$ and $U[t-1]$.
Thus $v_i[t]$ is computed as a 
weighted average of values in $[\mu[t-1],U[t-1]]$, and, therefore,
it will also be in $[\mu[t-1],U[t-1]]$.
\end{itemize}
Since $\forall i\in\scriptv-\scriptf$, $v_i[t]\in [\mu[t-1],U[t-1]]$,
the validity condition (\ref{e_validity}) is satisfied.
\end{proof}

\dividerline

\begin{lemma}
\label{lemma:psi}
Consider node $i\in\scriptv-\scriptf$.
Let $\psi\leq \mu[t-1]$. Then, for $j\in \{i\}\cup N_i^*[t]$,
\[
v_i[t] - \psi  \geq  a_i ~ (w_j - \psi)
\]
Specifically, for fault-free $j\in \{i\}\cup N_i^*[t]$,
\[
v_i[t] - \psi  \geq  a_i ~ (v_j[t-1] - \psi)
\]
\end{lemma}
\begin{proof}
In (\ref{e_Z}), for each $j\in N_i^*[t]$, consider two cases:
\begin{itemize}
\item Either $j=i$ or  $j\in N_i^*[t]\cap (\scriptv-\scriptf)$: Thus, $j$ is fault-free.
In this case, $w_j=v_j[t-1]$. Therefore,
$\mu[t-1] \leq w_j\leq U[t-1]$.
\item $j$ is faulty: In this case, $f$ must be non-zero (otherwise,
all nodes are fault-free).  
From Corollary~\ref{cor:2f+1}, $|N_i^-|\geq 2f+1$.
Then it follows that, in step 2 of Algorithm 1, the smallest $f$
values in $r_i[t]$ contain the state of at least one fault-free node,
say $k$.
This implies that $v_k[t-1] \leq w_j$.
This, in turn, implies that
$\mu[t-1] \leq w_j.$
\end{itemize}
Thus, for all $j\in \{i\}\cup N_i^*[t]$, we have $\mu[t-1] \leq w_j$.
Therefore,
\begin{eqnarray}
w_j-\psi\geq 0 \mbox{\normalfont~for all~} j\in\{i\} \cup N_i^*[t]
\label{e_algo_1}
\end{eqnarray}
Since weights in Equation~\ref{e_Z} add to 1, we can re-write that equation
as,
\begin{eqnarray}
v_i[t] - \psi &=& \sum_{j\in\{i\}\cup N_i^*[t]} a_i \, (w_j-\psi) \\
\nonumber
&\geq& a_i\, (w_j-\psi), ~~\forall j\in \{i\}\cup N_i^*[t]  ~~~~~\mbox{\normalfont from (\ref{e_algo_1})}
\end{eqnarray}
For non-faulty $j\in \{i\}\cup N_i^*[t]$, $w_j=v_j[t-1]$, therefore,
\begin{eqnarray}
v_i[t] -\psi &\geq & a_i\, (v_j[t-1]-\psi)
\end{eqnarray}
\end{proof}

\dividerline

\begin{lemma}
\label{lemma:Psi}
Consider node $i\in\scriptv-\scriptf$.
Let $\Psi\geq U[t-1]$. Then, for $j\in \{i\}\cup N_i^*[t]$,
\[
\Psi - v_i[t] \geq  a_i ~ (\Psi - w_j)
\]
Specifically, for fault-free $j\in \{i\}\cup N_i^*[t]$,
\[
\Psi - v_i[t] \geq  a_i ~ (\Psi - v_j[t-1])
\]
\end{lemma}
The proof of Lemma~\ref{lemma:Psi} is similar to that of
Lemma~\ref{lemma:psi}. The proof is presented in Appendix \ref{append:Psi}.

\dividerline

The lemma below uses parameter $\alpha$ defined in (\ref{e_alpha}).

\begin{lemma}
\label{lemma:bounded_value}
At time $s$ (i.e., at the end of the $s$-th iteration), suppose that
the fault-free nodes in $\scriptv-\scriptf$ can be partitioned into
non-empty sets
$R,L$ such that (i) $R$ propagates to $L$ in $l$ steps,
and (ii) the states of nodes
in $R$ are confined to an interval of length $\leq \frac{U[s]-\mu[s]}{2}$.
Then, 
\begin{eqnarray}
U[s+l]-\mu[s+l] \leq \left(1-\frac{\alpha^l}{2}\right)(U[s] - \mu[s])
\label{e:convergence:1}
\end{eqnarray}
\end{lemma}

\begin{proof}
Since $R$ propagates to $L$, as 
per Definition~\ref{def:absorb_sequence},
there exist sequences of sets
$R_0,R_1,\cdots,R_l$ and $L_0,L_1,\cdots,L_l$, where
\begin{itemize}
\item $R_0=R$, $L_0=L$, $L_l=\Phi$, for $0\leq \tau<l$, $L_\tau \neq \Phi$, and
\item for $0\leq \tau\leq l-1$,
\begin{list}{}{}
\item[*] $R_\tau\Rightarrow L_\tau$,
\item[*] $R_{\tau+1} = R_\tau\cup in(R_\tau\Rightarrow L_\tau)$, and
\item[*] $L_{\tau+1} = L_\tau - in(R_\tau\Rightarrow L_\tau)$
\end{list}
\end{itemize}
Let us define the following bounds on the states of the nodes
in $R$ at the end of the $s$-th iteration:
\begin{eqnarray}
M & = & max_{j\in R}~ v_j[s] \\ \label{e_M}
m & = & min_{j\in R}~ v_j[s] \label{e_m}
\end{eqnarray}
By the assumption in the statement of Lemma~\ref{lemma:bounded_value},
\begin{eqnarray}
M-m\leq \frac{U[s]-\mu[s]}{2} \label{e_M_m}
\end{eqnarray}
Also, $M\leq U[s]$ and $m\geq \mu[s]$.
Therefore, $U[s]-M\geq 0$ and $m-\mu[s]\geq 0$.

The remaining proof of Lemma~\ref{lemma:bounded_value} relies
on derivation of the three intermediate claims below. 

\shortdividerline

\begin{claim}
\label{claim:1}
For $0\leq \tau\leq l$, for each node $i\in R_\tau$,
\begin{eqnarray}
v_i[s+\tau] - \mu[s] \geq  \alpha^{\tau}(m-\mu[s])
\label{e_ind_1}
\end{eqnarray}
\end{claim}
\noindent{\em Proof of Claim \ref{claim:1}:}
The proof is by induction.

{\em Induction basis:}
For some $\tau$, $0\leq \tau< l$, for each node $i\in R_\tau$,
(\ref{e_ind_1}) holds.
By definition of $m$, the induction basis holds true for $\tau=0$.

\noindent{\em Induction:}
Assume that the induction basis holds true for some $\tau$, $0\leq \tau<l$.
Consider $R_{\tau+1}$.
Observe that $R_\tau$ and $R_{\tau+1}-R_\tau$ form a partition of $R_{\tau+1}$;
let us consider each of these sets separately.
\begin{itemize}
\item Set $R_\tau$: By assumption, for each $i\in R_\tau$, (\ref{e_ind_1})
holds true.
By validity of Algorithm 1, $\mu[s] \leq \mu[s+\tau]$.
Therefore, setting $\psi=\mu[s]$ in Lemma~\ref{lemma:psi},
we get,
\begin{eqnarray*}
v_i[s+\tau+1] - \mu[s] & \geq 
& a_i~(v_i[s+\tau] - \mu[s]) \\
& \geq & a_i~ \alpha^{\tau}(m-\mu[s]) ~~~~~~~~ \mbox{due to (\ref{e_ind_1})} \\
& \geq & \alpha^{\tau+1}(m-\mu[s])  ~~~~~~~~~~ \mbox{due to (\ref{e_alpha})}
\end{eqnarray*}

\item Set $R_{\tau+1}-R_\tau$: Consider a node $i\in R_{\tau+1}-R_\tau$. By definition
of $R_{\tau+1}$, we have that $i\in in(R_\tau\Rightarrow L_\tau)$.
Thus,
\[ |N_i^- \cap R_\tau| \geq f+1 \] 
In Algorithm 1, $2f$ values ($f$ smallest and $f$ largest) received by
node $i$ are eliminated before $v_i[s+\tau+1]$ is computed at
the end of $(s+\tau+1)$-th iteration. Consider two possibilities:
\begin{itemize}
\item Value received from one of the nodes in $N_i^- \cap R_\tau$ is
{\bf not} eliminated. Suppose that this value is received from
fault-free node $p\in N_i^-\cap R_\tau$. Then, by an argument similar to the
previous case, we can set $\psi=\mu[s]$
in Lemma~\ref{lemma:psi}, to obtain,
\begin{eqnarray*}
v_i[s+\tau+1] -\mu[s] & \geq & a_i~(v_p[s+\tau]-\mu[s]) \\
& \geq & a_i~ \alpha^{\tau}(m-\mu[s]) ~~~~~~~~ \mbox{due to (\ref{e_ind_1})} \\
& \geq & \alpha^{\tau+1}(m-\mu[s])  ~~~~~~~~~~ \mbox{due to (\ref{e_alpha})}
\end{eqnarray*}

\item Values received from {\bf all} (there are at least $f+1$) nodes
in $N_i^- \cap R_\tau$ are eliminated. Note that in this case $f$ must be
non-zero (for $f=0$, no value is eliminated, as already considered in the
previous case). By Corollary~\ref{cor:2f+1}, we know that
each node must have at least $2f+1$ incoming edges. 
Since at least $f+1$ values from nodes in $N_i^- \cap R_\tau$ are
eliminated, and there are at least $2f+1$ values to choose
from, it follows that the values that are {\bf not} eliminated\footnote{At
least one value received from the nodes
in $N_i^-$ is not eliminated, since there are $2f+1$ incoming
edges, and only $2f$ values are eliminated.}
are within the interval to which the values from $N_i^- \cap R_\tau$ belong.
Thus, there exists a node $k$ (possibly faulty) from whom node $i$ receives
some value $w_k$ -- which is not eliminated -- and 
a fault-free node $p\in N_i^- \cap R_\tau$ such that 
\begin{eqnarray}
v_p[s+\tau] &\leq & w_k \label{e_wk}
\end{eqnarray}
Then by setting $\psi=\mu[s]$ in Lemma~\ref{lemma:psi} we have
\begin{eqnarray*}
v_i[s+\tau+1] -\mu[s] & \geq & a_i~(w_k -\mu[s]) \\
& \geq & a_i~(v_p[s+\tau]-\mu[s]) ~~~~~~~~ \mbox{due to (\ref{e_wk})} \\
& \geq & a_i~ \alpha^{\tau}(m-\mu[s]) ~~~~~~~~ \mbox{due to (\ref{e_ind_1})} \\
& \geq & \alpha^{\tau+1}(m-\mu[s])  ~~~~~~~~~~ \mbox{due to (\ref{e_alpha})}
\end{eqnarray*}
\end{itemize}
\end{itemize}
Thus, we have shown that for all nodes in $R_{\tau+1}$,
\[
v_i[s+\tau+1] -\mu[s] 
\geq \alpha^{\tau+1}(m-\mu[s]) 
\]
This completes the proof of Claim \ref{claim:1}.

\shortdividerline

\begin{claim}
\label{claim:2}
For each node $i\in \scriptv-\scriptf$,
\begin{eqnarray}
v_i[s+l] - \mu[s] \geq  \alpha^{l}(m-\mu[s])
\label{e_ind_2}
\end{eqnarray}
\end{claim}
\noindent{\em Proof of Claim \ref{claim:2}:}

Notice that by definition, $R_l = \scriptv-\scriptf$. Then the proof follows by setting $\tau = l$ in the above Claim \ref{claim:1}.

\comment{=========old proof========
The proof is by induction.

{\em Induction basis:} 
For some $t$, where
$\tau\leq t< l$, for all $i\in R_\tau$, 
\begin{eqnarray}
v_i[s+t] - \mu[s] \geq \alpha^{t} (m-\mu[s]). \label{e_ind_2}
\end{eqnarray}
Due to Claim \ref{claim:1} above, this induction basis holds for $t=\tau$. 

{\em Induction:} Suppose that, for some $t<l$, for all
$i\in R_\tau$, $v_i[s+t]-\mu[s] \geq \alpha^{t} (m-\mu[s])$
% ++++++++ why does this matter? ++++
% The validity of Algorithm 1 implies that $v_i[s+t]\geq \mu[s]$
% for all fault-free nodes $i$.
% +++++++++++++++++++++++++++++++++++
Then by using $\psi=\mu[s]$ in Lemma~\ref{lemma:psi} as before,
\begin{eqnarray*}
v_i[s+t+1] - \mu[s] & \geq & a_i~(v_i[s+t] - \mu[s]) \\
& \geq & a_i~ \alpha^{t}(m-\mu[s]) ~~~~~~~~ \mbox{due to (\ref{e_ind_2})} \\
& \geq & \alpha^{t+1}(m-\mu[s])  ~~~~~~~~~~ \mbox{due to (\ref{e_alpha})}
\end{eqnarray*}
This concludes the proof of Claim \ref{claim:2}.
}

\shortdividerline

By a procedure similar to the derivation of Claim \ref{claim:2} above,
we can also prove the claim below. The proof of Claim \ref{claim:3}
is presented in the Appendix for completeness.
\begin{claim}
\label{claim:3}
For each node $i\in \scriptv-\scriptf$,
\begin{eqnarray}
U[s] - v_i[s+l] \geq  \alpha^{l}(U[s]-M)
\label{e_ind_3a}
\end{eqnarray}
\end{claim}

\shortdividerline

\noindent
Now let us resume the proof of the Lemma \ref{lemma:bounded_value}.
Note that $R_l=\scriptv-\scriptf$. Thus, 
\begin{eqnarray}
U[s+l] & = & \max_{i\in\scriptv-\scriptf}~ v_i[s+l] \nonumber \\
& \leq & U[s] - \alpha^{l}(U[s]-M) \mbox{~~~~~~~~~~~by (\ref{e_ind_3a})}
\label{e_U}
\end{eqnarray}
and
\begin{eqnarray}
\mu[s+l] & = & \min_{i\in\scriptv-\scriptf}~ v_i[s+l] \nonumber \\
& \geq & \mu[s] + \alpha^{l}(m-\mu[s]) \mbox{~~~~~~~~~~~by (\ref{e_ind_2}})
\label{e_mu}
\end{eqnarray}
Subtracting (\ref{e_mu}) from (\ref{e_U}),
\begin{eqnarray}
U[s+l]-\mu[s+l] & \leq & U[s] - \alpha^{l}(U[s]-M)  - \mu[s] - \alpha^{l}(m-\mu[s]) \nonumber \\
&=& (1-\alpha^l)(U[s]-\mu[s]) + \alpha^l(M-m) \\
&\leq& (1-\alpha^l)(U[s]-\mu[s]) + \alpha^l~\frac{U[s]-\mu[s]}{2}
		\mbox{~~~~~~~~~~by (\ref{e_M_m})} \\
&\leq& (1-\frac{\alpha^l}{2})(U[s]-\mu[s]) 
\end{eqnarray}
This concludes the proof of Lemma~\ref{lemma:bounded_value}.
\end{proof}

\dividerline

\begin{theorem}
\label{thm:convergence}
Suppose that $G(\scriptv, \scripte)$ satisfies Theorem~\ref{thm:nc}.
Then Algorithm 1 satisfies the {\em convergence} condition.
\end{theorem}
\begin{proof}
Our goal is to prove that, given any $\epsilon>0$, there
exists $\tau$ such that
\begin{equation}
U[t]-\mu[t] \leq \epsilon ~~~\forall t\geq \tau
\end{equation}

Consider $s$-th iteration, for some $s\geq 0$.
If $U[s]-\mu[s]=0$, then the algorithm has already converged, and the proof
is complete, with $\tau=s$.

Now consider the case when $U[s]-\mu[s]>0$.
Partition $\scriptv-\scriptf$ into two subsets, $A$ and $B$, such
that, for each node $i\in A$, 
$v_i[s]\in \left[\mu[s], \frac{U[s]+\mu[s]}{2}\right)$, and
for each node $j\in B$,
$v_j[s] \in \left[\frac{U[s]+\mu[s]}{2}, U[s]\right]$.
By definition of $\mu[s]$ and $U[s]$, there exist fault-free nodes
$i$ and $j$ such that $v_i[s]=\mu[s]$ and $v_j[s]=U[s]$.
Thus, sets $A$ and $B$ are both non-empty.
By Lemma \ref{lemma:must_absorb}, one of the following two conditions
must be true:
\begin{itemize}
\item Set $A$ propagates to set $B$. Then, define $L=B$ and $R=A$.
The states of all the nodes in $R=A$ are confined within an
interval of length
$<\frac{U[s]+\mu[s]}{2} - \mu[s] \leq \frac{U[s]-\mu[s]}{2}$.

\item Set $B$ propagates to set $A$. Then, define $L=A$ and $R=B$.
In this case, states of all the nodes in $R=B$ are confined within an interval of length
$\leq U[s]-\frac{U[s]+\mu[s]}{2} \leq \frac{U[s]-\mu[s]}{2}$. 

% \item There exist sets $L$ and $R$ such that $A\subseteq L$,
% $R\subseteq B$, sets $L,R$ form a partition of $\scriptv-\scriptf$,
% and $R$ propagates to $L$.
% Since $R\subseteq B$, the states in $R$ are confined to
% an interval no longer than the interval to which the states in $B$
% are confined. Thus, this interval is also $\leq \frac{U[s]-\mu[s]}{2}$. 

\end{itemize}
In both cases above, we have found non-empty sets $L$ and $R$
such that (i) $L,R$ is a partition of $\scriptv-\scriptf$,
(ii) $R$ propagates to $L$, and (iii) the states in $R$ are confined
to an interval of length $\leq \frac{U[s]-\mu[s]}{2}$.
Suppose that $R$ propagates to $L$ in $l(s)$ steps, where $l(s)\geq 1$.
By Lemma~\ref{lemma:bounded_value},
\begin{eqnarray}
U[s+l(s)]-\mu[s+l(s)] \leq \left( 1-\frac{\alpha^{l(s)}}{2}\right)(U[s] - \mu[s])
\label{e_t}
\end{eqnarray}
Since $n-f-1 \geq l(s)\geq 1$ and $0<\alpha\leq 1$, $0\leq \left( 1-\frac{\alpha^{l(s)}}{2}\right)<1$.

Let us define the following sequence of iteration indices\footnote{Without loss of generality, we assume that $U[\tau_i]-\mu[\tau_i] > 0$. Otherwise, the statement is trivially true due to the validity shown in Theorem \ref{thm:validity}.}:
\begin{itemize}
\item $\tau_0 = 0$,
\item for $i>0$, $\tau_i = \tau_{i-1} + l(\tau_{i-1})$, where $l(s)$ for any given $s$ was defined above.
\end{itemize}

By repeated application of the argument leading to (\ref{e_t}), we can prove
that, for $i\geq 0$,

\begin{eqnarray}
U[\tau_i]-\mu[\tau_i] \leq \left( \Pi_{j=1}^i\left( 1-\frac{\alpha^{\tau_i-\tau_{i-1}}}{2}\right)\right)~(U[0] - \mu[0])
\end{eqnarray}

For a given $\epsilon$,
by choosing a large enough $i$, we can obtain
\[
\left(\Pi_{j=1}^i\left( 1-\frac{\alpha^{\tau_i-\tau_{i-1}}}{2}\right)\right)~(U[0] - \mu[0]) \leq \epsilon
\]
and, therefore,
\begin{eqnarray}
U[\tau_i]-\mu[\tau_i] \leq  \epsilon
\end{eqnarray}
For $t\geq \tau_i$, by validity of Algorithm 1, it follows that
\[
U[t]-\mu[t] \leq
U[\tau_i]-\mu[\tau_i] \leq  \epsilon
\]
This concludes the proof.
\end{proof}

\section{Applications}
\label{s_application}

In this section, we use the results in the previous sections to examine whether
iterative approximate Byzantine consensus algorithm exists in some specific networks.

\subsection{Core Network}

Graph $G(\scriptv,\scripte)$ is said to be undirected iff $(i,j)\in\scripte$ implies that
$(j,i)\in\scripte$.
We now define a class of undirected graphs, named  {\em core network}.
\begin{definition}
\noindent{\em Core Network:}
A graph $G(\scriptv, \scripte)$ consisting of $n>3f$ nodes is said to be
a {\em core network} if the following properties are satisfied:
(i) it includes a clique formed by nodes in $K\subseteq \scriptv$,
such that $|K|= 2f+1$, as a subgraph
and, (ii) each node $i\not\in K$ has links to all the nodes in $K$.
That is,
(i) $\forall ~i, j \in K, (i, j) \in \scripte$ and $(j,i)\in\scripte$, and (ii) $\forall~v \in \scriptv - K$, and $\forall~u \in K$, $ (v, u) \in \scripte$ and $(u,v)\in\scripte$.
\end{definition}

It is easy to show that a core network satisfies the necessary condition in Theorem~\ref{thm:nc}.
Therefore, Algorithm 1 achieves approximate consensus in such network. We conjecture that
a core network with $n=3f+1$ has the smallest number of edges possible in any undirected
network of $3f+1$ nodes for which an iterative approximate consensus algorithm exists.

\subsection{Hypercube}

If the conensus algorithms are {\bf not} required to satisfy the consatraints imposed
on iterative algorithms in this paper, then
it is known that conensus can be achieved in undirected graphs with connectivity $>2f$ \cite{AA_nancy}.
However, connectivity of $2f+1$ by itself is not sufficient for iterative algorithms of interest
in this paper. For example, a $d$-dimensional binary hypercube is an undirected graph
consisting of $2^d$ nodes and has connectivity $d$. However, a cut of this graph that removes
edges along any one dimension fails to satisfy the necessary condition in
Theorem~\ref{thm:nc}, since each node has exactly one edge that belongs to the cut.
Thus, each node in one part of the partition is neighbor to fewer than $f+1$ nodes in
the other part, for any $f\geq 1$.
Figure~\ref{f_cube} illustrates such a partition for a 3-dimensional binary cube.
Each undirected link $(i,j)$ in the figure represents two directed
edges, namely, $(i,j)$ and $(j,i)$.

\begin{figure}
\centering
\includegraphics[width=0.6\textwidth]{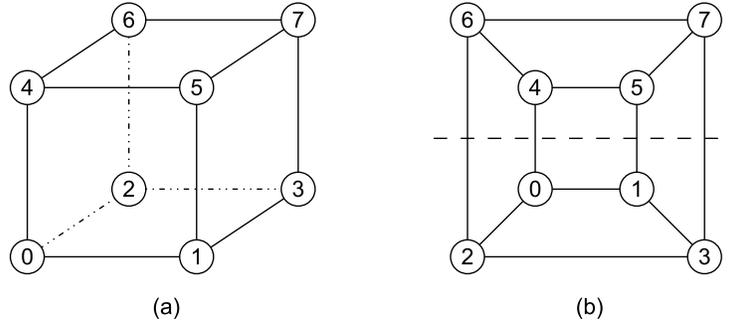}
\caption{(a) 3-dimensional cube. (b) 3-dimensional cube redrawn to illustrate
the partitions \{0,1,2,3\} and \{4,5,6,7\}.}
\label{f_cube}
\end{figure}

\subsection{Chord Network}

% It is somewhat surprising that directed and symmetric network has to be larger and denser in order to satisfy Theorem~\ref{thm:nc}. Here, we show that Chord graph (\cite{chord}) needs more nodes to tolerate faults via two examples. The particular type of Chord we explore is defined as:

A {\em chord} network is a directed graph defined as follows. This network is similar but not identical to the network in \cite{chord}.

\begin{definition}
\noindent{\em Chord network:}
A graph $G(\scriptv, \scripte)$ consisting of $n>3f$ nodes is said to be a chord network 
if (i) $\scriptv=\{0,1,\cdots,n-1\}$, (ii) $\forall i\in\scriptv$, $(i,j)\in\scripte$
iff $j=i+k\mod n$, where $1\leq k\leq 2f+1$. That is,
for each node $i\in\scriptv$, $(i, (i+1)\mod n), (i, (i+2) \mod n), ..., (i, (i+2f+1) \mod n) \in E$.
\end{definition}
The case when $f = 1$ and $n = 4$ results in a fully connected graph,
which trivially satisfies Theorem~\ref{thm:nc}.
The following results can be shown for two other specific chord networks:
\begin{itemize}
\item When $f = 2$ and $n = 7$, the chord network does not satisfy Theorem~\ref{thm:nc}.

 Let $\scriptv = \{0, 1, ..., 6\}$. Then the counter example is as follows:

Let node $5, 6$ be faulty. Then consider $L = \{0, 2\}$ and $R = \{1, 3, 4\}$. This partition fails Theorem~\ref{thm:nc}. Obviously, $L \not\Rightarrow R$, since $|L| < f+1 = 3$. However, $R \not\Rightarrow L$, since $N_0^- \cap R = \{3, 4\}$ and $N_2^- \cap R = \{1, 4\}$, which have size less than $3$. Notice that this example also illustrate that connectivity of $2f+1$ by itself is not sufficient in an directed and symmetric network.

\item The Chord network with $f = 1$ and $n = 5$ satisfies Theorem~\ref{thm:nc}.
\end{itemize}

\section{Asynchronous Networks}

The above results can be generalized to derive necessary and sufficient condition for (totally) asynchronous network under which the algorithm defined in \cite{AA_Dolev_1986} would work correctly. In essence, the primary change is that the requirement of $\geq f+1$ incoming links in the definition of $\Rightarrow$ needs to be replaced by $\geq 2f+1$ links. This implies that $|N_i^-|\geq 3f+1$ for each node $i$ when $f>0$ and $n$, number of nodes, must exceed $5f$. The above results can also be generalized to the (partially) asynchronous model defined in Section 7 of \cite{AA_convergence_markov} that allows for message delay of up to $B$ iterations.

Full details of the above generalizations will be presented in a future technical report. 

\section{Conclusion}

This paper proves a necessary and sufficient condition for the existence of
iterative approximate consensus algorithm in arbitrary directed graphs. As a special case,
our results can also be applied to undirected graphs. We also use the necessary
and sufficient condition to determine whether such iterative algorithms
exist for certain specific graphs.

In our ongoing research, we are exploring extensions of the above results
by relaxing some of the assumptions made in this work. 

% \bibliography{IBA}

%\newpage
\appendix
%\appendixpage

\section{Proof of Claim \ref{claim:3}}
\label{s_appendix}

In this section, we will prove the claim \ref{claim:3} in Section \ref{s_sufficiency}:

For each node $i\in \scriptv-\scriptf$,

\begin{eqnarray*}
U[s] - v_i[s+l] \geq  \alpha^{l}(U[s]-M)
\end{eqnarray*}

\begin{proof}

Similar to the proof of claim \ref{claim:2}, we will first prove the following claim:

\begin{claim}
\label{claim:3-1}
For $0\leq \tau\leq l$, for each node $i\in R_\tau$,
\begin{eqnarray}
U[s] - v_i[s+\tau] \geq  \alpha^{\tau}(U[s]-M)
\label{e_ind_3_1}
\end{eqnarray}
\end{claim}
\noindent{\em Proof of Claim \ref{claim:3-1}:}
The proof is by induction.

{\em Induction basis:}
For some $\tau$, $0\leq \tau< l$, for each node $i\in R_\tau$,
(\ref{e_ind_3_1}) holds.
By definition of $M$, the induction basis holds true for $\tau=0$.

\noindent{\em Induction:}
Assume that the induction basis holds true for some $\tau$, $0\leq \tau<l$.
Consider $R_{\tau+1}$.
Observe that $R_\tau$ and $R_{\tau+1}-R_\tau$ form a partition of $R_{\tau+1}$;
let us consider each of these sets separately.
\begin{itemize}
\item Set $R_\tau$: By assumption, for each $i\in R_\tau$, (\ref{e_ind_3_1})
holds true.
By validity of Algorithm 1, $U[s] \geq U[s+\tau]$.
Therefore, setting $\Psi=U[s]$ in Lemma~\ref{lemma:Psi},
we get,
\begin{eqnarray*}
U[s] - v_i[s+\tau+1] & \geq 
& a_i~(U[s] - v_i[s+\tau]) \\
& \geq & a_i~ \alpha^{\tau}(U[s]-M) ~~~~~~~~ \mbox{due to (\ref{e_ind_3_1})} \\
& \geq & \alpha^{\tau+1}(U[s]-M)  ~~~~~~~~~~ \mbox{due to (\ref{e_alpha})}
\end{eqnarray*}

\item Set $R_{\tau+1}-R_\tau$: Consider a node $i\in R_{\tau+1}-R_\tau$. By definition
of $R_{\tau+1}$, we have that $i\in in(R_\tau\Rightarrow L_\tau)$.
Thus,
\[ |N_i^- \cap R_\tau| \geq f+1 \] 
In Algorithm 1, $2f$ values ($f$ smallest and $f$ largest) received by
node $i$ are eliminated before $v_i[s+\tau+1]$ is computed at
the end of $(s+\tau+1)$-th iteration. Consider two possibilities:
\begin{itemize}
\item Value received from one of the nodes in $N_i^- \cap R_\tau$ is
{\bf not} eliminated. Suppose that this value is received from
fault-free node $p\in N_i^-\cap R_\tau$. Then, by an argument similar to the
previous case, we can set $\Psi=U[s]$
in Lemma~\ref{lemma:Psi}, to obtain,

\begin{eqnarray*}
U[s] - v_i[s+\tau+1] & \geq & a_i~(U[s] - v_p[s+\tau]) \\
& \geq & a_i~ \alpha^{\tau}(U[s] - M) ~~~~~~~~ \mbox{due to (\ref{e_ind_3_1})} \\
& \geq & \alpha^{\tau+1}(U[s] - M)  ~~~~~~~~~~ \mbox{due to (\ref{e_alpha})}
\end{eqnarray*}

\item Values received from {\bf all} (there are at least $f+1$) nodes
in $N_i^- \cap R_\tau$ are eliminated. Note that in this case $f$ must be
non-zero (for $f=0$, no value is eliminated, as already considered in the
previous case). By Corollary~\ref{cor:2f+1}, we know that
each node must have at least $2f+1$ incoming edges. 
Since at least $f+1$ values from nodes in $N_i^- \cap R_\tau$ are
eliminated, and there are at least $2f+1$ values to choose
from, it follows that the values that are {\bf not} eliminated
are within the interval to which the values from $N_i^- \cap R_\tau$ belong.
Thus, there exists a node $k$ (possibly faulty) from whom node $i$ receives
some value $w_k$ -- which is not eliminated -- and 
a fault-free node $p\in N_i^- \cap R_\tau$ such that 
\begin{eqnarray}
v_p[s+\tau] &\geq & w_k \label{e_wk2}
\end{eqnarray}
Then by setting $\Psi=U[s]$ in Lemma~\ref{lemma:Psi} we have

\begin{eqnarray*}
U[s] - v_i[s+\tau+1] & \geq & a_i~(U[s] - w_k) \\
& \geq & a_i~(U[s] - v_p[s+\tau]) ~~~~~~~~ \mbox{due to (\ref{e_wk2})} \\
& \geq & a_i~ \alpha^{\tau}(U[s]-M) ~~~~~~~~ \mbox{due to (\ref{e_ind_3_1})} \\
& \geq & \alpha^{\tau+1}(U[s]-M)  ~~~~~~~~~~ \mbox{due to (\ref{e_alpha})}
\end{eqnarray*}
\end{itemize}
\end{itemize}
Thus, we have shown that for all nodes in $R_{\tau+1}$,

\[
U[s] - v_i[s+\tau] \geq  \alpha^{\tau+1}(U[s]-M)
\]

This completes the proof of Claim \ref{claim:3-1}.

\shortdividerline

Now, we are able to prove Claim \ref{claim:3}.

\noindent{\em Proof of Claim \ref{claim:3}:}

Notice that by definition, $R_l = \scriptv-\scriptf$. Then the proof follows by setting $\tau = l$ in the above Claim \ref{claim:3-1}.

\end{proof}

% ========================================= lemma 2 =================

\section{Completing the proof of Lemma~\ref{lemma:must_absorb}}
\label{app:lemma:must_absorb}

The last line in the proof of Lemma~\ref{lemma:must_absorb} claims that:

``Since $B_k\subseteq B_0 = B$, $A\subseteq A_k$, and $B_k$ propagates
to $A_k$, it should be easy to see that $B$ propagates to $A$.''

We now prove the correctness of this claim.

\begin{proof}

Recall that $A_i$ and $B_i$ form a partition of $\scriptv-F$.

Let us define $P=P_0=B_k$ and $Q=Q_0=A_k$. Thus, $P$ propagates to $Q$.
Suppose that $P_0,P_1,...P_m$ and $Q_0,Q_1,\cdots,Q_m$ are
the propagating sequences in this case, with $P_i$ and $Q_i$ forming
a partition of $P\cup Q = A_k\cup B_k=\scriptv-F$.

Let us define $R=R_0=B$ and $S=S_0=A$.
Note that $R,S$ form a partition of $A\cup B=\scriptv-F$.
Now, $P_0=B_k\subseteq B=R_0$ and $S_0=A\subseteq A_k =Q_0$.
Also, $R_0-P_0$ and $S_0$ form a partition of $Q_0$.
% We will define sets additional sets $P_i,Q_i,R_i,S_i$ below.
Figure~\ref{f_appendixB} illustrates some of the sets used in this
proof.

\begin{figure}[h]
\centering
\includegraphics[width=0.7\textwidth]{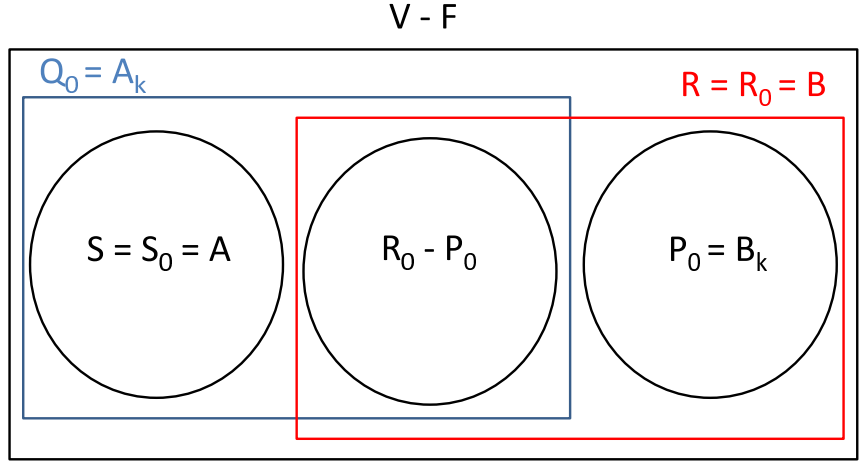}
\caption{Illustration for the proof of the last line in Lemma \ref{lemma:must_absorb}. In this figure, $R_0 = P_0 \cup (R_0 - P_0)$ and $Q_0 = S_0 \cup (R_0 - P_0)$.}
\label{f_appendixB}
\end{figure}

\begin{itemize}
\item
Define $P_1 = P_0\cup (in(P_0\Rightarrow Q_0))$, and $R_1 = V - F - P_1 = Q_0 - (in(P_0\Rightarrow Q_0))$ Also,
$R_1 = R_0\cup (in(R_0\Rightarrow S_0))$, and $S_1 = V - F - R_1 = S_0 - (in(R_0\Rightarrow S_0))$.

Since $R_0-P_0$ and $S_0$ are a partition of $Q_0$,
the nodes in $in(P_0\Rightarrow Q_0)$ belong to one of these
two sets. Note that $R_0-P_0\subseteq R_0$.
Also, $S_0 \cap in(P_0\Rightarrow Q_0) \subseteq in(R_0\Rightarrow S_0)$.
Therefore, it follows that $P_1 = P_0\cup (in(P_0\Rightarrow Q_0))
\subseteq R_0\cup (in(R_0\Rightarrow S_0)) = R_1$.

Thus, we have shown that, $P_1\subseteq R_1$. Then it follows that
$S_1\subseteq Q_1$.

\item For $0\leq i<m$, let us define $R_{i+1}=R_i\cup in(R_i\Rightarrow 
S_i)$ and $S_{i+1} = S_i - in(R_i\Rightarrow S_i)$. Then following an
argument similar to the above case, we can inductively show that,
$P_i\subseteq R_i$ and $S_i\subseteq Q_i$.
Due to the assumption on the length of the propagating
sequence above, $P_m=P\cup Q=\scriptv-F$.
Thus, there must exist $r\leq m$, such that $R_r=\scriptv-F$
and, for $i<r$, $R_i\neq \scriptv-F$.

The sequences $R_0,R_1,\cdots,R_r$ and $S_0,S_1,\cdots,S_r$ form
propagating sequences, proving that $R=B$ propagates to $S=A$. 
\end{itemize}

\end{proof}

\section{Proof of Lemma~\ref{lemma:Psi}}
\label{append:Psi}

\begin{proof}
In (\ref{e_Z}), for each $j\in N_i^*[t]$, consider two cases:
\begin{itemize}
\item Either $j=i$ or  $j\in N_i^*[t]\cap (\scriptv-\scriptf)$: Thus, $j$ is fault-free.
In this case, $w_j=v_j[t-1]$. Therefore,
$\mu[t-1] \leq w_j\leq U[t-1]$.
\item $j$ is faulty: In this case, $f$ must be non-zero (otherwise,
all nodes are fault-free).
From Corollary~\ref{cor:2f+1}, $|N_i^-|\geq 2f+1$.
Then it follows that, in step 2 of Algorithm 1, the largest $f$
values in $r_i[t]$ contain the state of at least one fault-free node,
say $k$.
This implies that $v_k[t-1] \geq w_j$.
This, in turn, implies that
$U[t-1] \geq w_j.$
\end{itemize}
Thus, for all $j\in \{i\}\cup N_i^*[t]$, we have $U[t-1] \geq w_j$.
Therefore,
\begin{eqnarray}
\Psi - w_j\geq 0 \mbox{\normalfont~for all~} j\in\{i\} \cup N_i^*[t]
\label{e_algo_11}
\end{eqnarray}
Since weights in Equation~\ref{e_Z} add to 1, we can re-write that equation
as,
\begin{eqnarray}
\Psi - v_i[t] &=& \sum_{j\in\{i\}\cup N_i^*[t]} a_i \, (\Psi-w_j) \\
\nonumber
&\geq& a_i\, (\Psi-w_j), ~~\forall j\in \{i\}\cup N_i^*[t]  ~~~~~\mbox{\normalfont from (\ref{e_algo_11})}
\end{eqnarray}
For non-faulty $j\in \{i\}\cup N_i^*[t]$, $w_j=v_j[t-1]$, therefore,
\begin{eqnarray}
\Psi-v_i[t] &\geq & a_i\, (\Psi-v_j[t-1])
\end{eqnarray}
\end{proof}

\end{document}